%% file: main.tex
\documentclass[journal]{IEEEtran}
\ifCLASSINFOpdf

\else

\fi

\usepackage{times}
\usepackage{fancyhdr,graphicx,amsmath,amssymb}
\usepackage{subfigure}
\usepackage{color}
\usepackage{multirow}
\usepackage{amsthm}
\newtheorem{theorem}{Theorem}

\usepackage[ruled,linesnumbered]{algorithm2e}
\include{pythonlisting}
\usepackage[center]{caption}

\begin{document}
\title{An Efficient Network Solver for Electromagnetic Transient Simulation of Power Systems Based on Hierarchical Inverse Computation and Modification} 

\author{Lu~Zhang,
        Bin~Wang,
        Vivek Sarin,
        Weiping~Shi,
        P.~R.~Kumar,
        and~Le~Xie\\ Texas A\&M University\\College Station, TX, 77843, USA




\thanks{Preferred email for correspondence: Weiping Shi, wshi@tamu.edu.}
}

\maketitle

\begin{abstract}
\input{Sections/0_Abstract.tex}
\end{abstract}


\IEEEpeerreviewmaketitle

\section{Introduction}
\input{Sections/1_Introduction.tex}
\section{Brief Review of EMT Simulation}
\input{Sections/2_Review_of_EMT_simulation.tex}

\section{Hierarchical Approximation of $G$ Inverse}
\input{Sections/4_Hierarchical_G_inverse_approximation}

\section{Fast $G$ Inverse Modification}
\input{Sections/5_Application_in_EMT_Fault_Simulation}

\section{Case Study}
\input{Sections/6_Case_study}




\section{Concluding Remarks}
\input{Sections/7_Conclusion_and_Future_work.tex}


\ifCLASSOPTIONcaptionsoff
  \newpage
\fi


\bibliographystyle{IEEEtran}
\bibliography{References.bib}








\end{document}

%% file: Sections/0_Abstract.tex
In both power system transient stability and electromagnetic transient (EMT) simulations, up to 90\% of the computational time is devoted to solve the network equations, i.e., a set of linear equations. Traditional approaches are based on sparse LU factorization, which is inherently sequential. In this paper, EMT simulation is considered and an inverse-based network solution is proposed by a hierarchical method for computing and store the approximate inverse of the conductance matrix. The proposed method can also efficiently update the inverse by modifying only local sub-matrices to reflect changes in the network, e.g., loss of a line. Experiments on a series of simplified 179-bus Western Interconnection demonstrate the advantages of the proposed methods.

%% file: Sections/1_Introduction.tex

Transient stability (TS) and electromagnetic transient (EMT) simulations has been utilized in many critical applications in modern power industry, e.g., system dynamic security assessment, device insulation design and validation of control strategy \cite{2010milano:book}\cite{2014Ametani:book}. In practice, simulation speed is always a key aspect of simulation tools, especially for real-time applications and large-scale systems \cite{2017rahman}. As the penetration level of inverter-based generations, such as wind and solar, continues to increase in the upcoming decades, the complexity and dimension of the underlying systems will drastically increase along with more dynamic risks, e.g., sub-synchronous oscillations. This poses a challenge to the inevitable need for fast simulations of large-scale systems \cite{2019kroposki}.

Two major types of dynamic simulation in power industry are TS simulation, and EMT simulation. In both types, solving the network equations, i.e.,  a set of linear equations, is the most time-consuming step, taking up to $90\%$ of the overall computational time \cite{2019LJ:journal}\cite{2017Chen:journal}\cite{2019wang:sas}. The most popular approach is based on LU factorization \cite{2012Davis}\cite{2010Davis}, which is adopted by many commercial simulation tools. However, this approach is sequential in nature, making it difficult to be applied to parallel computing platforms. 

In contrast, the inverse-based approach exhibits great potential in speeding up the network solution by parallel matrix-vector multiplication. Fig. \ref{fig:solver} compares the LU-based and inverse-based approaches where black texts are our previous effort \cite{zhang2019:isgt,Lu}. In \cite{zhang2019:isgt}\cite{Lu}, we explored a fast direct inverse-based network solver, where the solving step outperforms the LU-based approach by $2.8 \times$ for large networks ($>$1000 buses) in runtime, even when sequentially computed. However, it requires significant time and memory to calculate and store the inverse of the large conductance matrix $G$. 

\begin{figure}[t]
	\centering
	\includegraphics[width=0.9\columnwidth]{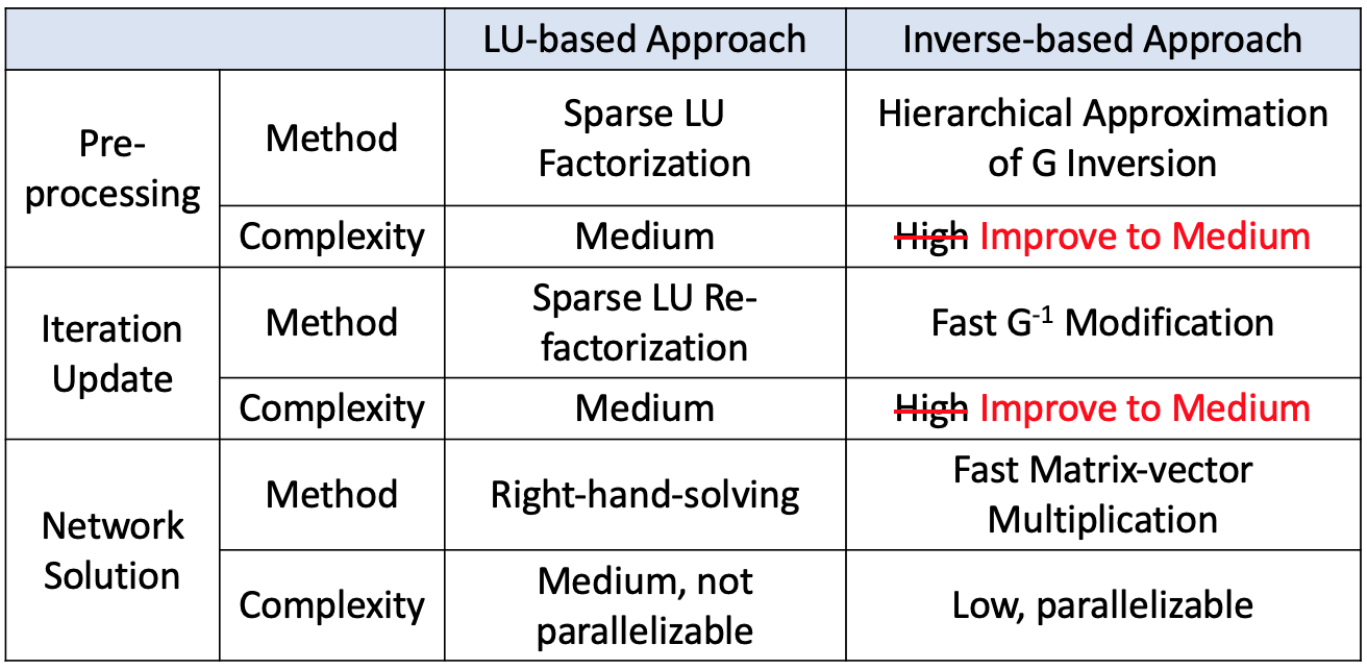}
	\caption{Comparison of LU-based and inverse-based approaches, where red texts are contributions of this paper.}
	\label{fig:solver}
\end{figure}

To further develop the technique of parallelizable network solvers, this paper proposes a fast and memory-efficient hierarchical approximation method for calculating and store the inverse of the large admittance/conductance matrices established from power grids. Different from traditional Diakoptic \cite{Diakoptics} and conventional hierarchical approaches \cite{intro}\cite{matrixInversion}, the proposed method decomposes the entire network into several sub-networks, then implicitly computes matrix inverse by using the inversions of sub-matrices hierarchically. A fast modification approach is further presented based on the inverted matrix by modifying only local entries and sub-matrices to reflect changes in network, e.g., loss of a line. The advantage of the proposed method is demonstrated only on EMT simulation in this paper, while similar results are expected for TS simulation.





The key contributions of this paper are as follows. 1) An efficient hierarchical method for calculating and store $G^{-1}$, which can be used for the inverse based solution of linear equations involved in large-scale power system TS and EMT simulations; 2) A fast and memory-efficient $G^{-1}$ modification approach to handle scenarios where $G$ changes, which may be used for simulating fault and topology changes.

%% file: Sections/2_Review_of_EMT_simulation.tex

In EMT simulation, the network equations are a set of linear equations, as shown in \eqref{eq:neteq}, which are derived from the nodal formulation and discretized by the numerical integration \cite{1981emtp:book}.


\begin{equation}
	G\mathbf{v}(t) = \mathbf{i}_{\textrm{in}}(t) + \mathbf{i}_{\textrm{his}}(t-\Delta t), \label{eq:neteq}
\end{equation}

\noindent where $G$ is the conductance matrix, $\mathbf{v}(t)$ is the vector of nodal voltages, $\mathbf{i}_{\textrm{in}}(t)$ is the vector of current injections, and $\mathbf{i}_{\textrm{his}}(t-\Delta t)$ is the vector of historical currents, and all values are real. 

The traditional solution to \eqref{eq:neteq} is LU-based which requires sequential forward backward substitute in the solving stage. This paper focuses on a fully parallelizable inverse-based approach which rewrites \eqref{eq:neteq} into a matrix-vector multiplication:
\begin{equation}\label{eq:neteq2}
    \begin{aligned}
        &\mathbf{v}(t) = G^{-1} \mathbf{i}(t),  
        \mathbf{i}(t) =  \mathbf{i}_{\textrm{in}}(t) + \mathbf{i}_{\textrm{his}}(t-\Delta t).
    \end{aligned}
\end{equation}

Our past efforts \cite{Lu} shows that, once $G^{-1}$ is calculated, the network solution time can be reduced from  O($N^2$) to O($N \log N$), where $N$ is the system size, even with sequential computation, and parallelizable schemes may be employed to further speed up the simulation.

%% file: Sections/4_Hierarchical_G_inverse_approximation.tex
Unlike $G$, $G^{-1}$ cannot be formed directly from the network. Even though $G^{-1}$ can be found via inversion of $G$, explicit inversion is very expensive. Thus, this section introduces a hierarchical approximation approach for computing $G^{-1}$ with much lower time and storage requirements.

\subsection{Approximation of $G$ Inverse}



Assuming $G$ and $G^{-1}$ are symmetric. Express $G$ and $G^{-1}$ in a block matrix format according to a partition of network $G$ into sub-network $G_1$ and $G_2$,

\begin{equation}\label{eq:Gaprx1}
    G = \begin{bmatrix} G_{1} & H \\ H^{T} & G_{2}
    \end{bmatrix}
    ~\text{and}~{G}^{-1} = \begin{bmatrix} A & M \\ M^{T} & B
    \end{bmatrix},
\end{equation}
where the sizes of $G_{1}$ is $m \times m$, $G_{2}$ is $n \times n$, $H$ is $m \times n$, and $m/3$, $n/3$ are the number of buses represented by $G_1$ and $G_2$, and each bus has three phases.

To compute the approximate $G^{-1}$, named as $\widetilde{G}^{-1}$, we want to find a matrix making the off-diagonals of $ G  \widetilde{G}^{-1}$ equal to 0, and the diagonals close to the identity matrix. Therefore, we wish to restrict sub-matrices $A$ and $B$ as follows,
\begin{equation}\label{eq:Gaprx2}
\widetilde{G}^{-1} = \begin{bmatrix} G_{1}^{-1} & M \\ M^{T} & G_{2}^{-1} 
\end{bmatrix}.
\end{equation}
We then find $M$ and $M^T$.
\begin{equation}\label{eq:gg}
    \begin{aligned}
        G \times \widetilde{G}^{-1} 
        &= \begin{bmatrix} G_{1} & H \\ H^{T} & G_{2}
        \end{bmatrix} \begin{bmatrix} G_{1}^{-1} & M \\ M^{T} & G_{2}^{-1} \end{bmatrix}\\
        &= \begin{bmatrix} I + MH^{T} & G_{1}^{-1}H + MG_{2} \\ G_{2}^{-1}H^{T} + M^{T}G_{1} & I + M^{T}H
        \end{bmatrix}.
    \end{aligned}
\end{equation}
Forcing the off-diagonals to be 0, we get
\begin{equation}\label{eq:Meq}
    \left\{\begin{matrix}
    G_{1}^{-1}H + MG_{2} = 0,  \\ 
    G_{2}^{-1}H^{T} + M^{T}G_{1} = 0. 
    \end{matrix}\right.
\end{equation}
Thus,
\begin{equation}\label{eq:M}
    M = - G_{1}^{-1} H G_{2}^{-1}.
\end{equation}
Plug \eqref{eq:M} into \eqref{eq:Gaprx2},
\begin{equation}\label{eq:Ginv1}
    \widetilde{G}^{-1} = \begin{bmatrix} G_{1}^{-1} & - G_{1}^{-1} H G_{2}^{-1} \\ - G_{2}^{-1} H^{T} G_{1}^{-1} & G_{2}^{-1} 
    \end{bmatrix}.
\end{equation}

Note that the approximation inverse approach works for both real and complex matrices, while in this paper, real matrix is considered since the conductance matrix involved in nodal formulation based EMT simulation is real-valued.

In power systems, matrix $H_{m\times n}$ in $G$ is determined by the network, i.e., the transmission lines connecting two groups of buses $s_1$ and $s_2$, where $s_1$ contains $m$ buses, and $s_2$ contains $n$ buses. Suppose there are $k$ transmission lines between $s_1$ and $s_2$, then $H = h_{1} h_{2}^{T}$,
where $h_{1}$ is an $m\times k$ matrix representing the interaction from $s_1$ to $s_2$, and $h_{2}$ is an $n\times k$ matrix representing the interaction from $s_2$ to $s_1$. In general, $k$ is much less than $m$ and $n$. We can rewrite $M$ as follows.
\begin{equation}\label{eq:MH}
    M = - (G_{1}^{-1} h_{1}) (G_{2}^{-1} h_{2})^{T}.
\end{equation}
If we multiply the matrices in the order given by \eqref{eq:MH} instead of \eqref{eq:M}, the computation cost will be reduced to $km^2 + kn^2 + kmn$. In addition, minimizing $k$ will further reduce the time.



\subsection{Hierarchical Approximation of $G$ Inverse}

To further reduce the time to compute $G^{-1}$, we propose the hierarchical approximation, which is illustrated in Fig. \ref{fig:HG}. 

\begin{figure}[h]
	\centering
	\includegraphics[width=0.9\columnwidth]{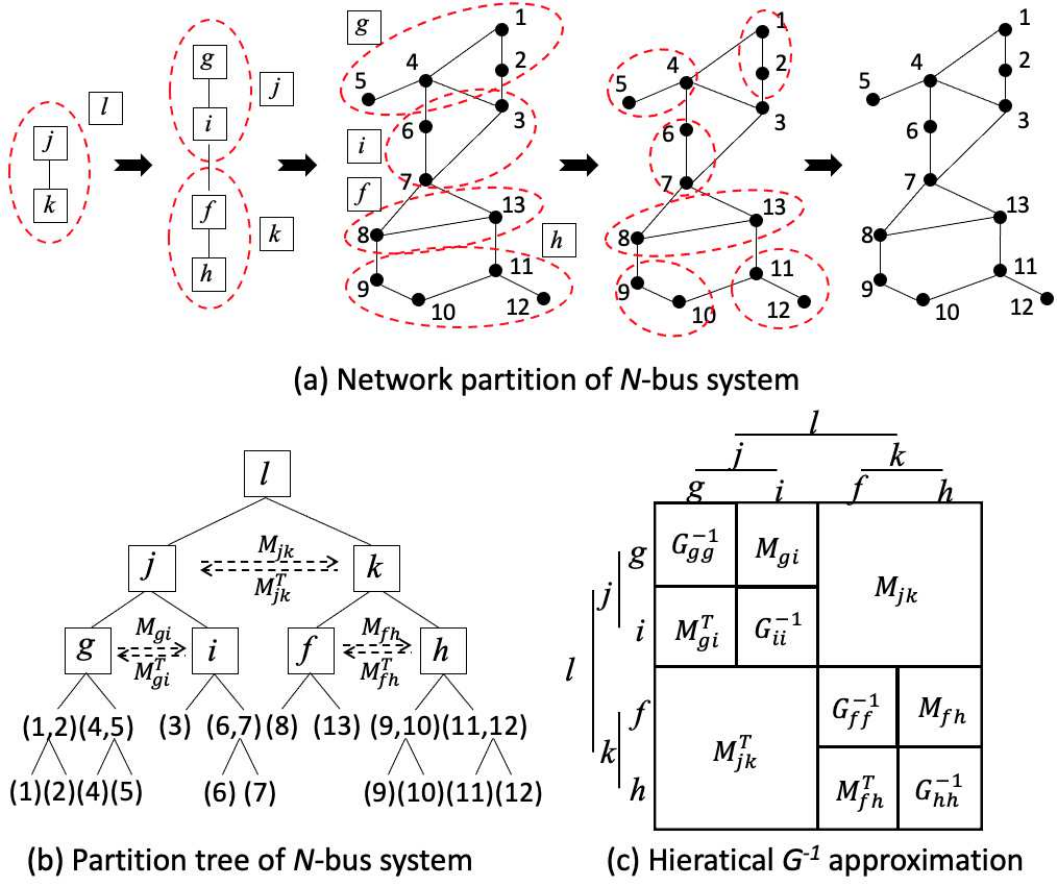} 
	\caption{Hierarchical approximation for calculating $G^{-1}$.}
	\label{fig:HG}
\end{figure}

The process is as described. First, we recursively partition the network into multiple sub-networks. In Fig. \ref{fig:HG} (a), network $l$ is first bi-partitioned into sub-networks $j$ and $k$, then $j$ and $k$ are further bi-partitioned into sub-network $g$, $i$, and $f$, $h$, respectively, so on so forth until each sub-network contains only one bus. The partition thereby builds a binary tree where the root represents the whole network, the internal node represents a group of buses, and the leaf represents an individual bus, shown as Fig. \ref{fig:HG} (b). Correspondingly, the network partition divides $G^{-1}$ into a hierarchy, where each sub-matrix, including $G^{-1}$ itself, is further divided into 4 sub-matrices, until each sub-matrix is a single element. In Fig. \ref{fig:HG} (c), group $j$ consists of 7 buses and group $k$ consists of 6 buses, thus the interaction within $j$ is expressed by a $7 \times 7$ sub-matrix $G_{jj}^{-1}$, where the row and column indices correspond to buses in $j$; the interaction within $k$ is expressed by a $6 \times 6$ sub-matrix $G_{kk}^{-1}$, where the row and column indices correspond to buses in $k$, and the interaction between $j$ and $k$ is expressed by a $7 \times 6$ sub-matrix $M_{jk}$, where the row indices correspond to buses in $j$, and the column indices correspond to buses in $k$. Based on the hierarchically partitioned $G^{-1}$, we apply the approximate inverse approach recursively, eventually construct the entire $G^{-1}$ without any direct matrix inversion, i.e., $G_{jj}^{-1}$ is constructed by $G_{gg}^{-1}$, $G_{ii}^{-1}$ and $M_{gi}$, and $M_{gi}$ is calculated by \eqref{eq:MH}, where $G_{gg}^{-1}$ is a $6 \times 6$ matrix, $G_{ii}^{-1}$ is a $3 \times 3$ matrix, $h_{gi}$ is a $6 \times 3$ matrix, and $h_{ig}$ is a $3 \times 3$ matrix.


To improve the accuracy of the approximation, we pre-define a node threshold $d_{th}$ to limit the size of the minimum group of buses. If a sub-network contains fewer than $d_{th}$ buses, the partition stops. For the example in Fig. \ref{fig:pruning}, we set $d_{th}$ as 4, thus node $g$, $i$, $f$, and $h$ are the minimal sub-network where the partition stops. To achieve high approximation accuracy and low computation cost, the network partition algorithm needs to minimize the edges between two sub-networks. In this paper, we adopted the partition algorithm used in \cite{Lu}.


\begin{figure}[h]
	\centering
	\includegraphics[width=0.95\columnwidth]{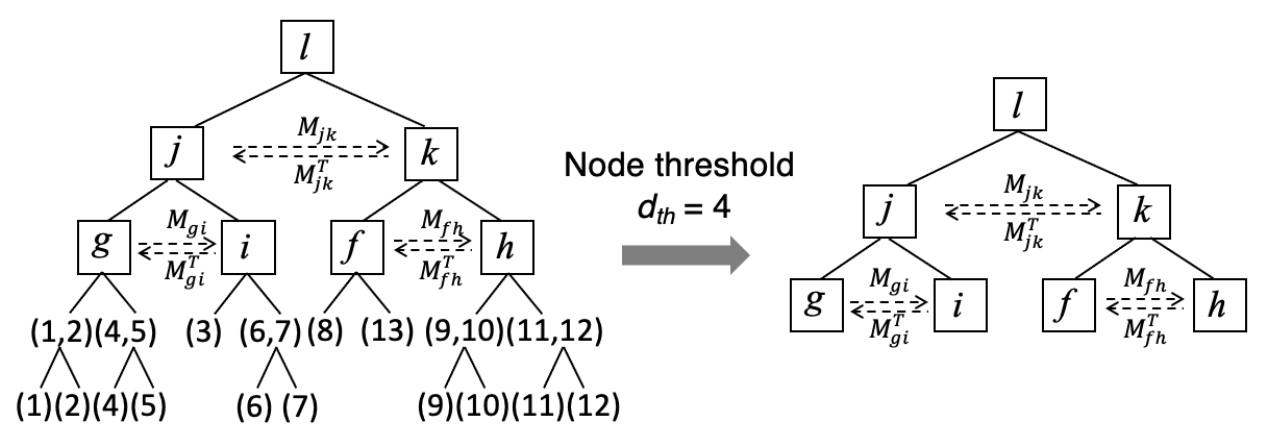} 
	\caption{Partition tree of $G^{-1}$ by node threshold $d_{th}$ = 4.}
	\label{fig:pruning}
\end{figure}

The algorithm is presented below. The input is the root of the partition tree $i$, with its left child being $l$ and right child being $r$. Each tree node stores the bus indices of the sub-network. The outputs are $M_{lr}$ and the approximation of $G^{-1}_{ii}$.

\begin{algorithm}
\SetAlgoLined

\uIf{ i \textup{has less than} $d_{th}$ \textup{buses} }{
    direct compute the $G_{ii}^{-1}$; 
  }
  \Else{
    $G_l^{-1}$ = \text{Hierarchical\_Ginv}($l$);\\
    $G_r^{-1}$ = \text{Hierarchical\_Ginv}($r$);\\
    $H = G_{lr}$;
    Decompose $H = h_l * h_r^T$, where $h_l$ is an $l \times k$ matrix, $h_r$ is an $r \times k$ matrix, $k$ is the number of edges between $l$ and $r$; \\ 
    $M_{lr} = - (G_l^{-1} h_l)(G_r^{-1} h_r)^T$,
    $G_{ii}^{-1}=
    \begin{bmatrix} 
     G_l^{-1} & M_{lr} \\
     M^T_{lr}      & G_r^{-1} 
    \end{bmatrix}$;
  }
\KwRet $M_{lr}, G_{ii}^{-1}$

\caption{Hierarchical\_Ginv($i$)}
\label{alg:hiarGinv}
\end{algorithm}

The benefits of this approach are 1) Save computation of matrix inversion without direct invert $G$. In the algorithm, we only need to direct invert the diagonal matrices of each size bounded by node threshold $d_{th}$; 2) Reduce storage requirement. Instead of storing the entire dense matrix $G^{-1}$, this approach only requires to store diagonal matrices.


%% file: Sections/5_Application_in_EMT_Fault_Simulation.tex
This section introduces a fast $G^{-1}$ modification approach based on Algorithm \ref{alg:hiarGinv} to handle network changes.

\subsection{Modification Algorithm}

Take the network in Section III.B as an example. Suppose a fault occurs between bus 4 and bus 6. Denote $G_{pre}$, $G_{on}$ as conductance matrices in the pre-fault stage and fault-on stage, respectively. The construction of $G_{pre}$, $G_{on}$, $G_{pre}^{-1}$ and $G_{on}^{-1}$ is illustrated in Fig. \ref{fig:app}.

In the pre-fault stage, $G_{pre}$ is constructed according to the network before the fault occurs, shown in Fig. \ref{fig:app} (a). We hierarchically approximate $G_{pre}^{-1}$ by Algorithm \ref{alg:hiarGinv}. Name the approximation of ${G_{pre}^{-1}}$ as $\widetilde{G_{pre}^{-1}}$, shown in Fig. \ref{fig:app} (c).

\begin{figure}[h]
	\centering
	\includegraphics[width=0.7\columnwidth]{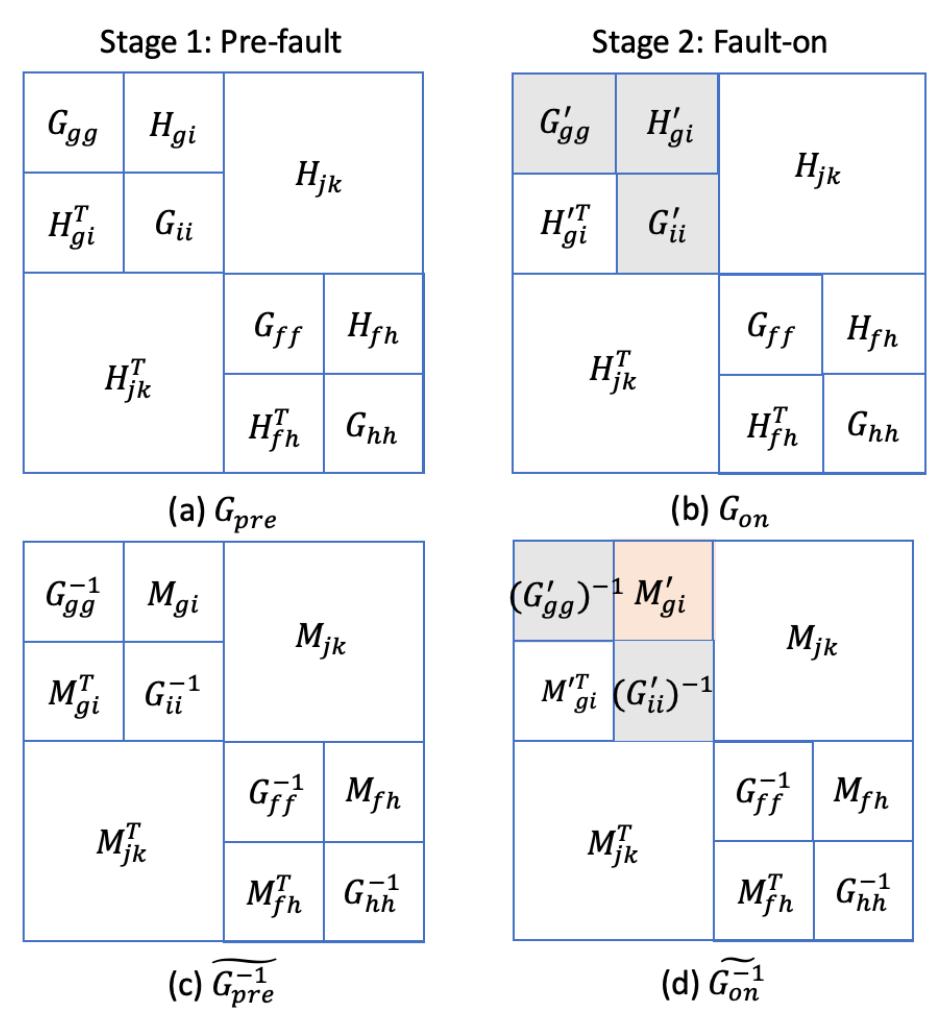}
	\caption{Modification of $G^{-1}$ under a fault between $g$ and $i$.}
	\label{fig:app}
\end{figure}

In the fault-on stage, $G_{pre}$ is updated to $G_{on}$ by modifying corresponding elements in $G_{gg}, G_{ii}$, and $H_{gi}$ due to the fault occurred between bus 4 and bus 6. Rename the modified block matrices as $G_{gg}^{'}, G_{ii}^{'}, H_{gi}^{'}$, shown in the Fig. \ref{fig:app} (b). Since only three blocks are modified, while others stay the same. According to Algorithm \ref{alg:hiarGinv}, $\widetilde{G^{-1}_{on}}$ is formed by modifying $(G_{gg}^{'})^{-1}, (G_{ii}^{'})^{-1}$, and $M_{gi}^{'}$ only, shown in Fig. \ref{fig:app} (d). 

As described, the modification of $G^{-1}$ updates a relative small portion of sub-blocks and avoids re-construct the entire $G^{-1}$. The blocks which need to be updated are easily located by the partition tree. In this example, bus 4 belongs to node $g$, and bus 6 belongs to node $i$. Starting from node $g$ and $i$, we trace up to their ancestors until reaching their lowest common ancestor (LCA) which is node $j$. Thus node $j$ is the smallest group that contains both bus 4 and 6, and the smallest group that is directly influenced by the fault. Therefore, the interaction within node $j$ needs to be modified, thus all matrices that describe interactions between nodes along the paths from $g$ to $j$ and from $i$ to $j$ need to be modified. Since node $g$ and $i$ are leaves, the interaction within node $g$ and $i$ which represented as $(G_{gg}^{'})^{-1}$ and $(G_{ii}^{'})^{-1}$ are computed directly by inversion. The interaction between node $g$ and $i$ represented as $M_{gi}^{'}$ is computed by \eqref{eq:M}. The interaction within node $j$ is approximated by interaction within node $g$, node $i$ and between node $g$ and $i$ by \eqref{eq:Gaprx2}. The tree with modification nodes is shown in Fig. \ref{fig:appT}, where the grey nodes are updated by directly inverse, the orange node is the LCA of grey nodes estimated by combining grey nodes.

\begin{figure}[h]
	\centering
	\includegraphics[width=0.5\columnwidth]{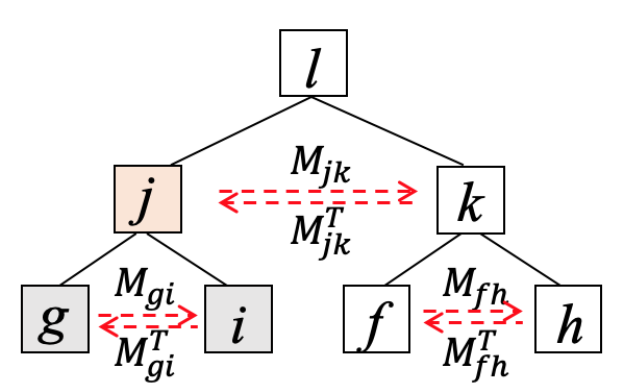} 
	\caption{The updated tree of $G^{-1}$ under a fault.}
	\label{fig:appT}
\end{figure}

The detailed algorithm is presented below. Suppose the network is partitioned and $G^{-1}$ is approximated by Algorithm \ref{alg:hiarGinv} at beforehand. The inputs are node $i$ and $j$, where node $i$ contains the fault bus $i$, node $j$ contains the fault bus $j$.



\begin{algorithm}
\SetAlgoLined

Find the LCA of node $i$ and $j$, name it as node $k$;\\
UpdateM($k$, $i$, $j$); 

\textbf{return} 
\caption{Modification($i, j$)}
\label{alg:modify}


  \SetKwFunction{FMain}{UpdateM}
  \SetKwProg{Fn}{Function}{:}{}
  \Fn{\FMain{$k, i, j$}}{
    \If{k \textup{is a leaf in the partition tree} }{
        \If{ $k$ \textup{is the same as} $i$ \textup{or} $j$}{
            direct compute the $G_{kk}^{-1}$;
        }
    \textbf{return} 
    }
    \text{UpdateM}(left child of $k$, $i$, $j$); \\
    \text{UpdateM}(right child of $k$, $i$, $j$); \\
    
    Update $M_{lr}^{-1}$ by \eqref{eq:MH}, where $l$ and $r$ are the left and right child of $k$, respectively; \\
    \KwRet\;
  }
  
\end{algorithm}





The computation cost of the fast $G^{-1}$ modification approach highly depends on the fault location. In the best case, the LCA of node $i$ and $j$ is itself, so that we don't need to update any $M$ matrices, just update the diagonal matrix locally. In the worst case, the LCA of node $i$ and $j$ is the root, so that besides updating the diagonal matrix locally, we also need to update multiple $M$ matrices from leaf to root.

\subsection{Time Complexity}

\begin{theorem}
Given an $N$-bus power system which is hierarchically partitioned, the time complexity of fast $G^{-1}$ modification algorithm based on hierarchical approximation for $G$ inversion is $O(n^2)$, where $n$ is the smallest group of buses which contain buses directly connected to the fault.
\end{theorem}

\begin{proof}
Suppose the $N$-bus power system network is partitioned into multiple groups, thereby produces a binary tree where the root represents the whole network, the leaves represents small groups of buses where each group contains no more than $d_{th}$ buses. Therefore, the partition tree has $O(N)$ nodes and the height of the tree is $O(\log N)$. 

Consider when a fault occurs between bus $i$ and $j$, which belongs to leaves $i$ and $j$ in the partition tree, respectively. Let node $k$ be the lowest common ancestor of leaves $i$ and $j$ and let the number of buses under node $k$ be $n$. According to Algorithm \ref{alg:modify}, the time complexity of fast $G^{-1}$ modification algorithm only depends on $n$. Let $T(n)$ be the time complexity of fast $G^{-1}$ modification algorithm under a fault, then we have
\begin{equation*}
    T(n) \leq  \left\{\begin{matrix}
    C & \mbox{if $n\leq d_{th}$} \\ 
    2T(\frac{n}{2}) + rn^2 & \mbox{otherwise},
    \end{matrix}\right.
\end{equation*}

\noindent where $C$ is a constant, $T\left( \frac{n}{2} \right)$ is the time complexity of updating diagonal matrices due to the partition, and $rn^2$ is the upper bound on time complexity of updating the off-diagonal matrices $M$, $r$ is the number of edges between two groups.

To solve the recurrence relation, we have
\begin{equation*}
    \begin{aligned}
    T(n) &\leq 2T\left( \frac{n}{2} \right) + rn^2 \\
        &= 2\left(2T\left(\frac{n}{4}\right) + r\left( \frac{n}{2}\right)^2\right) + rn^2
        = \cdots \\ 
        &= \frac{n}{d_{th}}T(d_{th}) + 
        rn^2  \sum_{i=0}^{\log (n/d_{th})} \frac{1}{2^i} \\
        &< \frac{n}{d_{th}} + 2rn^2 = O(n^2).
    \end{aligned}
\end{equation*}

The sum to $\log n/d_{th}$ is due to the fact that the length of the path from node $i$ to $k$, and $j$ to $k$ is at most $\log n/d_{th}$. 
 
\end{proof}

%% file: Sections/6_Case_study.tex
This section presents comprehensive case studies for EMT simulation on $179$-bus based systems to evaluate the performance of the proposed approaches. The data and one-line diagrams of the $179$-bus system can be found in \cite{2016wang:lib}. In the simulation, Algorithm \ref{alg:hiarGinv} is used to approximate $G^{-1}$. If a fault occurs, Algorithm \ref{alg:modify} is used to update $G^{-1}$. For the network solution, the low-rank approximation based approach is adopted to further compress $G^{-1}$, then fast matrix-vector multiplication \cite{Lu} is applied to solve the network equations. 

\subsection{Node threshold v.s. $G^{-1}$ Accuracy}

In this test, $G$ matrix for a 179-bus system is used for evaluating the impact of the node threshold $d_{th}$ on the accuracy of $G^{-1}$. Fig. \ref{fig:nodeERR} shows that as $d_{th}$ increases from 2 to 74, the error exponentially drops from $2.5\times10^{-3}$ to $1.8\times10^{-9}$. However, as $d_{th}$ increases, the computation cost increases. To balance the trade-off between error and computation cost, $d_{th}$ is set to 74 in our test cases.

\begin{figure}[h]
	\centering
	\includegraphics[width=0.7\columnwidth]{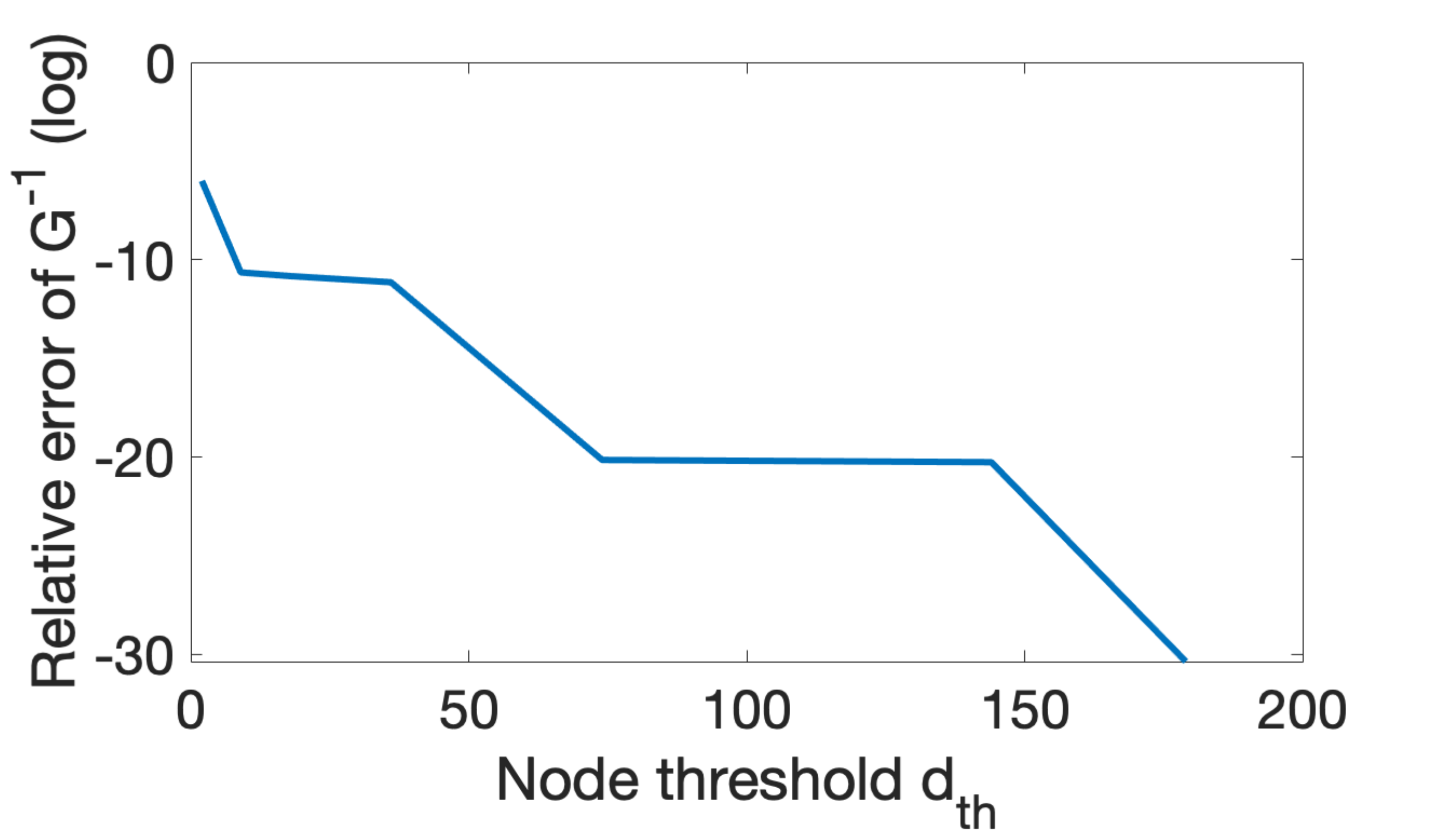} 
	\caption{Relative error of $G^{-1}$ with node threshold $d_{th}$.}
	\label{fig:nodeERR}
\end{figure}

\subsection{CPU time in terms of FLOPS}

In this test, a series of large systems consist of multiple copies of a 179-bus system interconnected as an array. A similar system is used in \cite{song2017:journal}. The FLOPS involved in approximate $G^{-1}$ by Algorithm \ref{alg:hiarGinv} and the network solution by fast matrix-vector multiplication \cite{Lu} at one time step are counted. The FLOPS count includes total numbers of floating point addition and multiplication.

\subsubsection{FLOPS for hierarchical approximation of $G$ inverse}
Traditionally, the time complexity of matrix inverse is O$(N^3)$, where $N$ is the system size. Fig. \ref{fig:gflop} shows that in the series of 179-bus based systems, the FLOPS of Algorithm \ref{alg:hiarGinv} is O$(N^2)$.

\subsubsection{FLOPS for network solution}
Fig. \ref{fig:nflop} shows that in the series of 179-bus based systems with up to $2148$ buses, the maximum FLOPS of LU-based approach is $5.56 \times 10^{6}$, while the fast matrix-vector multiplication \cite{Lu} is $1.43\times10^{6}$, leading to an $74\%$ reduction compared with LU-based approach. 

\begin{figure}[ht]
	\centering
	\subfigure[Flops for $G^{-1}$ computing ]{\label{fig:gflop}
	\includegraphics[width=0.44\columnwidth]{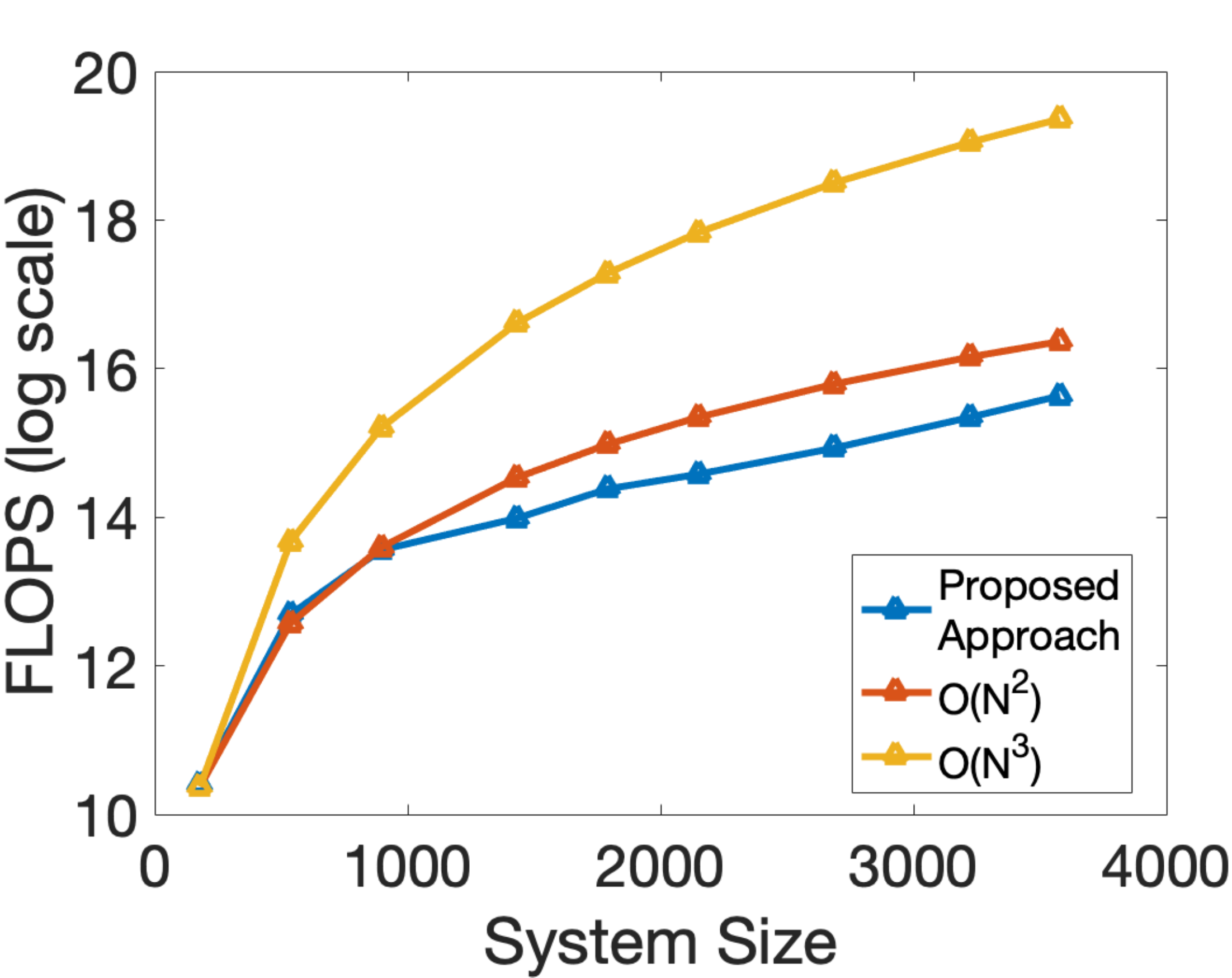}}~~~~
	\subfigure[Flops for network solution]{\label{fig:nflop}
	\includegraphics[width=0.44\columnwidth]{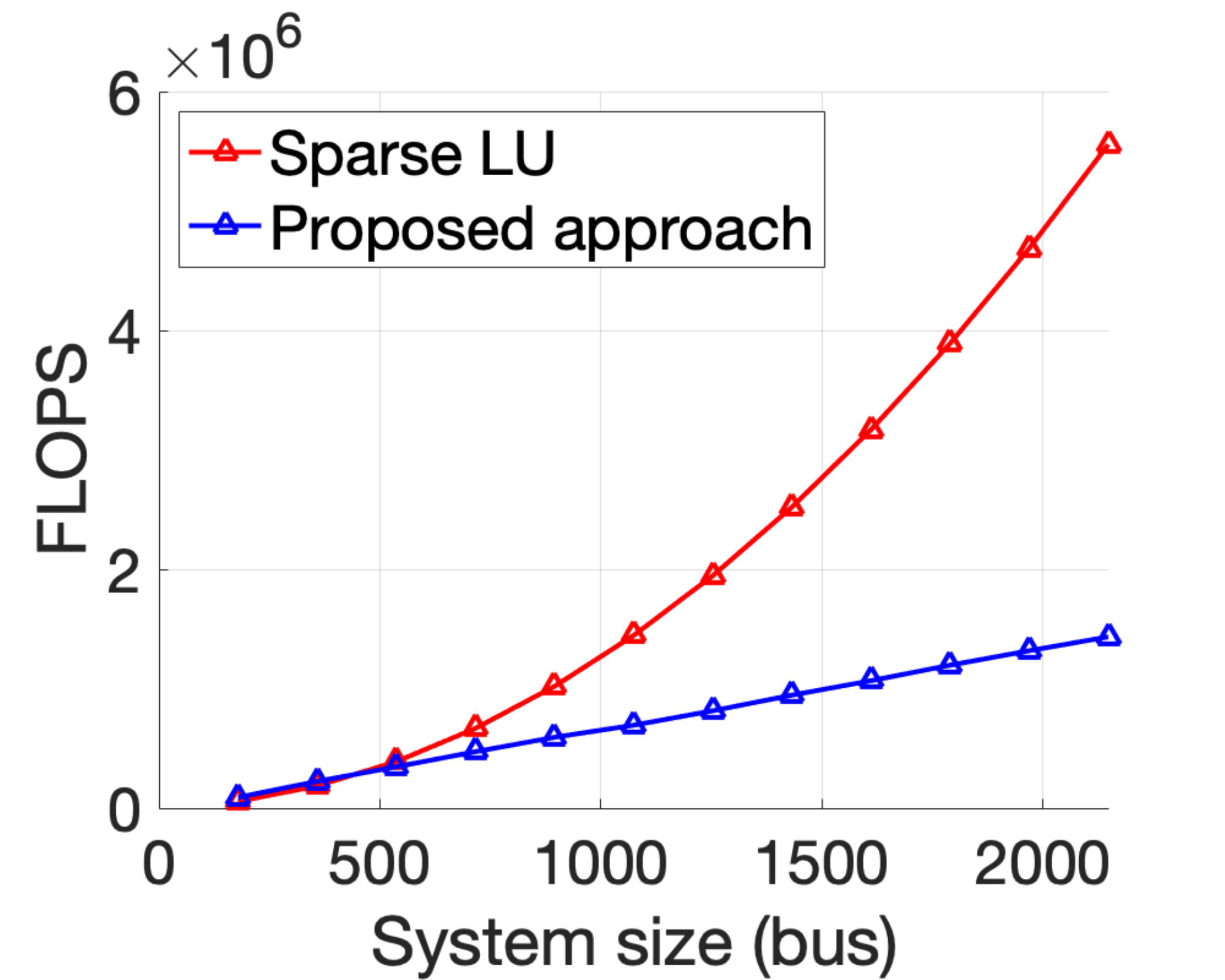}}\\
	\caption{FLOPS comparison.}\label{fig:flop}
\end{figure}

\subsection{EMT fault simulation}

In this test, a three-phase ground fault on a 179-bus system is simulated with a 20$\mu$s time step and a total simulation time length of 60ms. The fault between bus 1 and bus 81 is added at 10ms and lasts for 20ms. The fault resistance is 10 $\Omega$. The partial network with the fault is shown in Fig. \ref{fig:179sys}. 

\begin{figure}[h]
	\centering
	\includegraphics[width=0.7\columnwidth]{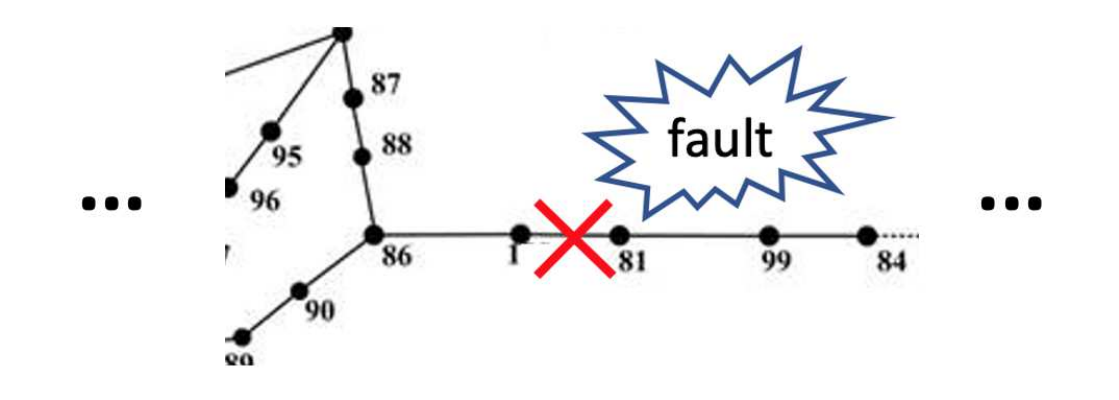}
	\caption{The 179-bus system with a fault.}
	\label{fig:179sys}
\end{figure}

\subsubsection{Accuracy}
 An LU-based network solver is used as a reference. The relative error of the bus voltage is measured at each simulation step. As shown in Fig. \ref{fig:simu} and \ref{fig:err}, the result by the proposed approach matches well with the reference. The maximal relative error of bus voltage is $7.4\times10^{-5}$.

\begin{figure}[ht]
	\centering
	\subfigure[LU-based Approach]{\label{fig:lu}
	\includegraphics[width=0.45\columnwidth]{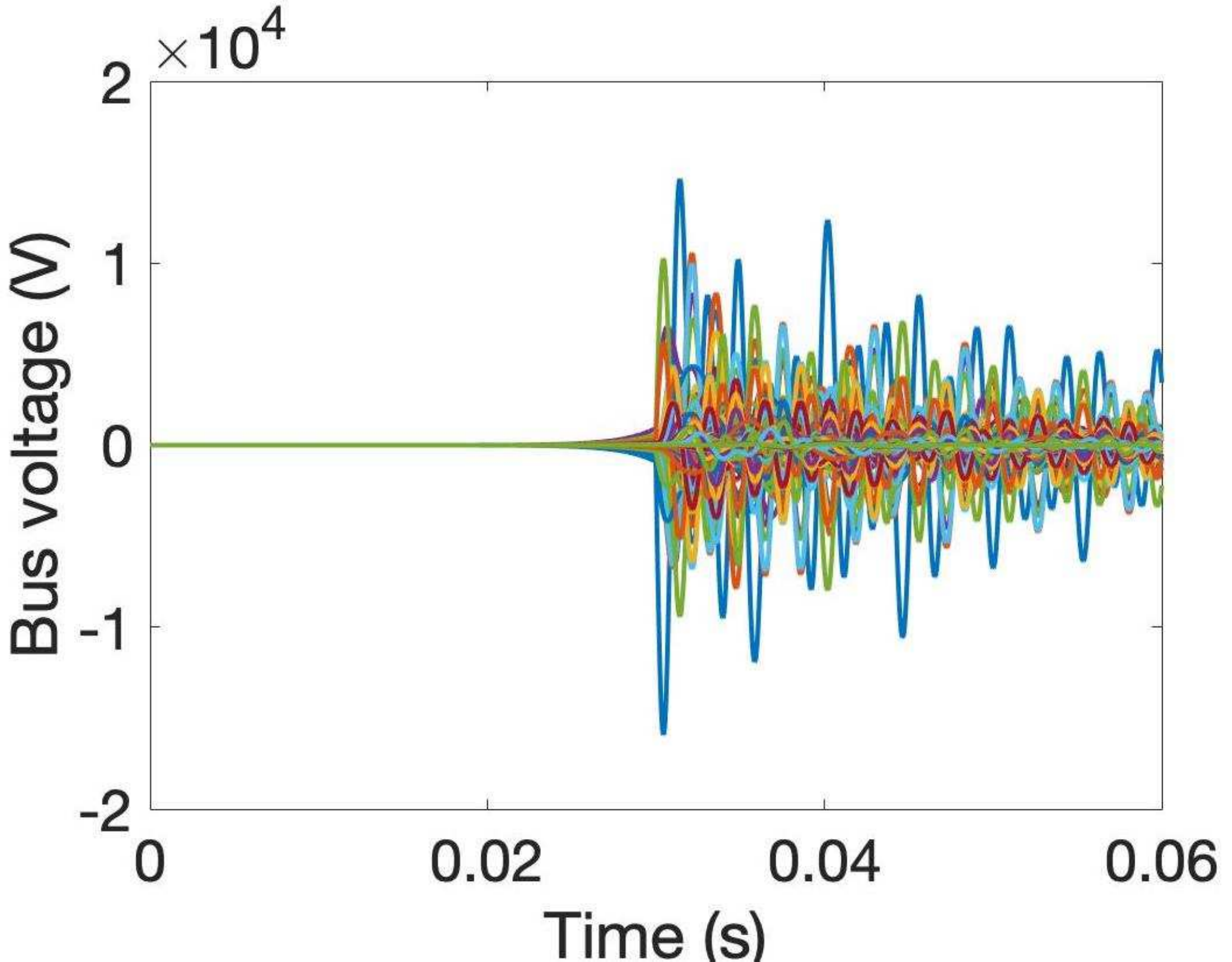}}~~~~
	\subfigure[Proposed approach]{\label{fig:propose}
	\includegraphics[width=0.45\columnwidth]{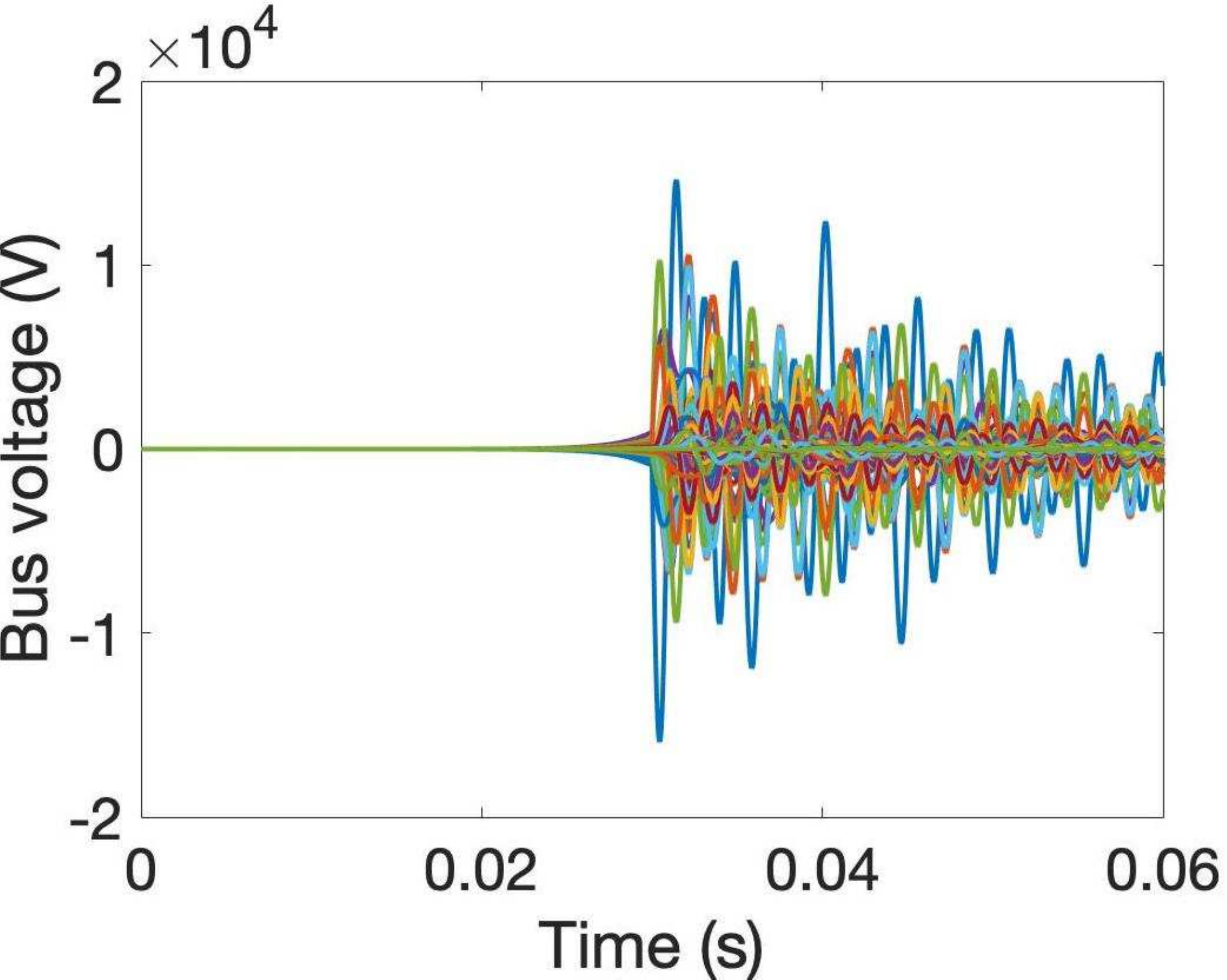}}\\
	\caption{Accuracy comparison.}\label{fig:simu}
\end{figure}

\begin{figure}[h]
	\centering
	\includegraphics[width=0.7\columnwidth]{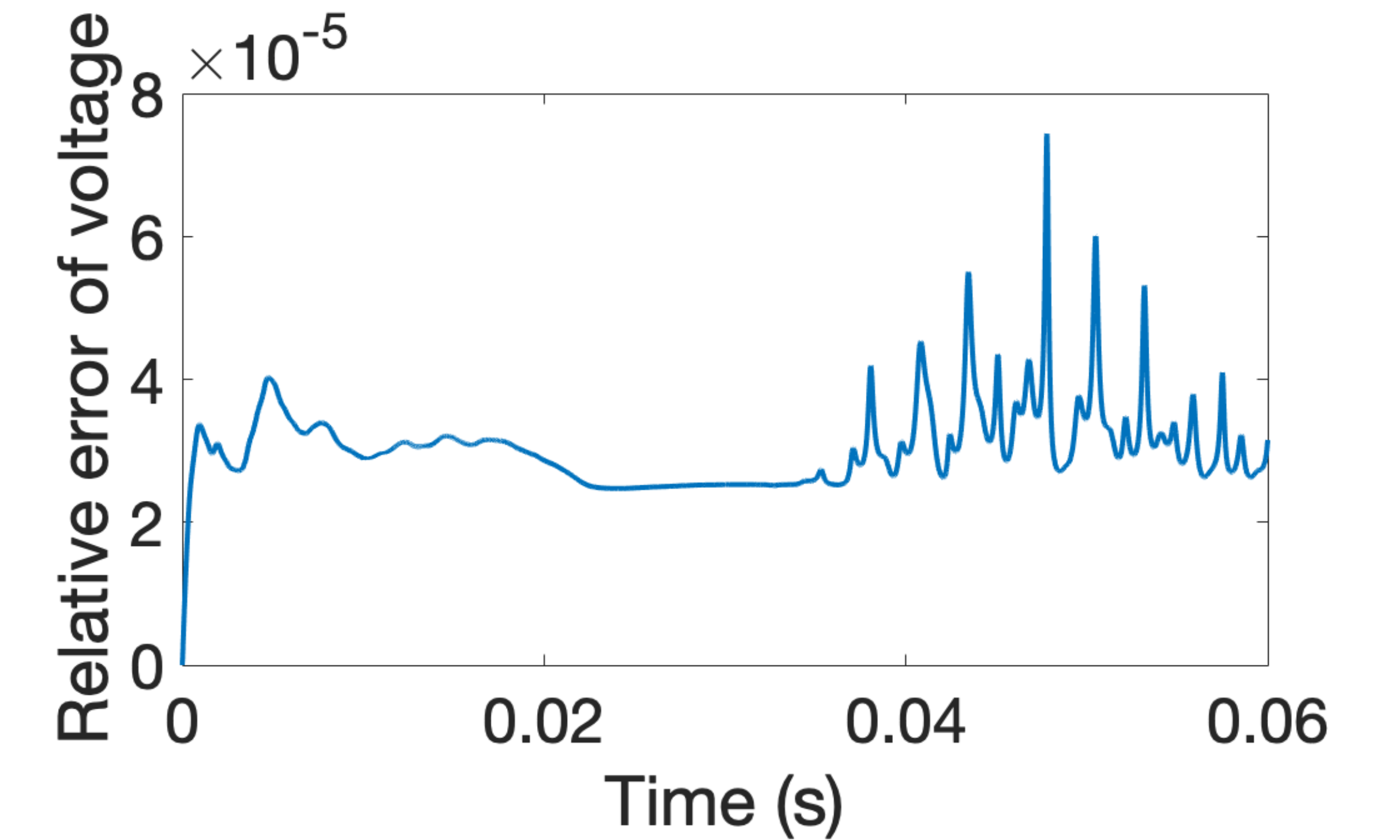} \caption{Relative error of bus voltage on 179-bus system.}
	\label{fig:err}
\end{figure}

\subsubsection{Computation cost for $G$ inverse modification}
The cost of fast $G^{-1}$ modification approach depends on the fault location. In this case, the fault occurs in the same leaf node of the partition tree. Since the node threshold is 74, the proposed approach only needs to invert the matrix with $74 \times 74$ elements instead of $179 \times 179$ elements, leading to 82.9\% reduction.

%% file: Sections/7_Conclusion_and_Future_work.tex
This paper presents an inverse-based scheme to solve the power network algebraic equations. To break down the bottleneck of computing matrix inversion, this paper proposes a hierarchical approximation of $G^{-1}$. Further, a fast modification approach is presented for facilitating the contingency analysis. The performance is demonstrated by the EMT simulation on large-scale systems created from a 179-bus base case. The experiments show that the proposed approaches exhibit great advantage in speeding up the network solution with high accuracy. Future work will investigate the scalability of this approach in larger realistic systems.